\documentclass[twocolumn]{autart}
\usepackage[dvips]{epsfig}

\usepackage{amsfonts}
  \usepackage{amsmath}
  \usepackage{graphicx}
  \usepackage{url}
  \usepackage{eepic, epsfig, amsmath, amssymb, latexsym, setspace, subfig}
  \usepackage{rotating}
  \usepackage{amsfonts}
  \usepackage{amsmath}
  \usepackage{graphicx}
  \usepackage{nicefrac}
\usepackage{lipsum,multicol}

  \newtheorem{theorem}{Theorem}[section]
  
\newtheorem{lemma}[theorem]{Lemma}
  \newtheorem{proposition}[theorem]{Proposition}

  \numberwithin{equation}{section}

  \newenvironment{proof}{\noindent{\bf Proof}\hspace*{1em}}{\hfill\qed\vspace{0.125in}}

\newcommand{\x}{\mathbf{x}}
  \newcommand{\y}{\mathbf{y}}
  
  \newcommand{\w}{\mathbf{w}}

\begin{document}
  \begin{frontmatter}

\title{Outliers and Random Noises in System Identification: a Compressed
Sensing Approach}

  \thanks{This paper was not presented in any IFAC meeting.
  This work was supported in
      part by NSF and DoE grants. The corresponding author ErWei Bai, 
     Tel:+3193355949 and Fax: +3193356028}

\author[First]{Weiyu Xu}
  \author[Second]{Erwei Bai}
  \author[Third]{Myung Cho}

  \address[First]{Dept. of Electrical and Computer Engineering \\
      University of Iowa, Iowa City, Iowa 52242}

  \address[Second]{Dept. of Electrical and Computer Engineering \\
      University of Iowa, Iowa City, Iowa 52242\\
    School of Electronics, Electrical Engineering and Computer Science
    \\
    Queen's University, Belfast, UK}
  \address[Third]{Dept. of Electrical and Computer Engineering \\
      University of Iowa, Iowa City, Iowa 52242}


\begin{abstract}
In this paper, we consider robust system identification under sparse
outliers and random noises. In this problem, system parameters are observed
through a Toeplitz matrix. All observations are subject to random noises and a few are
corrupted with outliers. We reduce this problem of system identification to a
sparse error correcting problem using a Toeplitz structured real-numbered coding matrix.
We prove the performance guarantee of Toeplitz structured matrix in sparse error
correction. Thresholds on the percentage of correctable errors for Toeplitz
structured matrices are established. When both outliers and observation noise are
present, we have shown that the estimation error goes to $0$ asymptotically as long as
the probability density function for observation noise is
not ``vanishing'' around $0$.
No probabilistic assumptions are imposed on the outliers.
\end{abstract}

  \end{frontmatter}

\section{Introduction}
In a linear system identification setting, an unknown system parameter vector
$\x \in R^m$ is often observed through a Toeplitz matrix $H \in R^{n \times m}$ ($n \geq m$), namely
\begin{equation*}
\y=H\x,
\end{equation*}
where $\y=(y_1,y_2,...,y_n)^T$ is the system output and
the Toeplitz matrix $H$ is equal to
\begin{equation}
\label{Top}
          \begin{bmatrix}
                h_{-m+2} & h_{-m+3} & \ldots & h_{1}\\
                h_{-m+3} & h_{-m+4} & \ldots & h_{2}\\
                 \vdots& \ldots & \ldots & \vdots  \\
                 \vdots& \ldots & \ldots & \vdots  \\
                 \vdots& \ldots & \ldots & \vdots  \\
                                h_{-m+n+1} & \ldots & \ldots & h_{n}
\end{bmatrix},
\end{equation}
with $h_i$, $-m+2\leq i \leq n$, being the system input sequence.

If there is no interference or noise in the observation $\y$, one can then simply
recover $\x$ from matrix inversion. However, in applications, the observations $\y$ are corrupted by noises and a few elements can be exposed to
large-magnitude gross errors or outliers. Such outliers can happen with the failure of measurement devices, measurement
communication errors and the interference of adversary parties.
 Mathematically, when both additive observation noise and outliers are present,
the observation $\y$ can be written as
\begin{equation}
\y=H\x+\e+\w,
\label{eq:model}
\end{equation}
where $\e$ is a sparse outlier vector with $k \ll n$ non-zero elements, and
$\w$ is a measurement noise vector with each element usually being assumed to be
i.i.d. random variables. We further assume $m$ is fixed and $n$ can increase, which is often the case in system identifications \cite{Ljung}.

If only random measurement errors are present, the least-square solutions generally
provide an asymptotically good estimate. However, the least-square estimate
breaks down
in the presence of outliers.
Thus, it is necessary to protect the estimates from both random noise and outliers.
Research along this direction has attracted a significant amount of attention,
for example, \cite{Bai,cook,Neter,Rousseeuw,Ljung,Stoica}.
An effective way is to visually inspect the residual plot and change
the obviously erroneous measurements ``by hand" to an appropriately interpolated
values \cite{Ljung}. The approach does not however always work. Another approach
was the idea of few violated constraints \cite{Bai} in the setting of the bounded
error parameter estimation. The other two popular methods in the statistical
literature to deal with the outliers are the least median squares and the least
trimmed squares \cite{Rousseeuw}. Instead of minimizing the sum of the residual
squares,
the least median squares yields the smallest value for the median of squared
residuals computed from the entire data set and the least trimmed
squares tries to minimize the sum of squared
residues over a subset of the given data. Both have shown robustness against
the outliers \cite{Rousseeuw}. The problem is their computational complexity.
Both algorithms are nonlinear and in fact combinatory in nature. This limits
their practical applications if $n$ and/or $m$ are not small or even modest. The most popular
way to deal with the outliers in the statistical literature is the
the least absolute deviation estimate ($\ell_1$ minimization) which has been
extensively studied
\cite{SIC,CandesErrorCorrection,CT1,DMTSTOC,CDC2011}. Instead of searching
for all the $\binom{n}{k}$ possibilities for the locations of outliers,
\cite{SIC,CandesErrorCorrection,CT1} proposed to minimize the least absolute
deviation:
\begin{equation}
\min_{\x}  \|\y-H\x\|_{1}.
\label{eq:errorcorrection}
\end{equation}
Under the assumption that the error $\e+\w$ is an i.i.d.
random sequence with a common density which has median zero and is continuous
and positive in the neighborhood of zero, the difference between the unknown $\x$
and its estimate is asymptotically Gaussian of zero mean \cite{SIC}.
The problem is that the assumption of i.i.d. of median zero on the unknown
outliers is very restrictive and
seldomly satisfied in reality.
We study the least absolute deviation estimator or $\ell_1$ minimization
from the compressed sensing
point of view and show that i.i.d. of median zero on the outliers are
unnecessary. In fact only the number of outliers relative to the total number
of data length plays a role.

%
Recovering signals from outliers or errors have been
studied \cite{CandesErrorCorrection,CT1,DMTSTOC,RudelsonVershynin,StojnicThresholds,XuHassibi,CDC2011}. In their setting, each element
of the nonsingular $(n-m) \times n$ matrix $A$ such that $AH=0$, is assumed to be i.i.d. random variables following a certain
distribution, for example, Gaussian distribution or Bernoulli distribution.
These types of matrices have been shown to obey certain conditions such as
restricted isometry conditions \cite{CandesErrorCorrection} so that
(\ref{eq:errorcorrection}) can correctly recover $\x$ when there are only
outliers present; and can recover $\x$ approximately when both outliers and
measurement noise exist. However, in the system identification problem, $H$ has
a natural Toeplitz structure and the elements of $H$ are not
independent but correlated. The natural
question is whether (\ref{eq:errorcorrection}) also provides performance
guarantee for recovering $\x$ with a Toeplitz matrix. We provide a positive
answer in this paper. 

Despite the fact that the elements of Toeplitz matrices are correlated,
we are able to show in this paper that Toeplitz structured matrices also enable the successful recovery of $\x$ by
using (\ref{eq:errorcorrection}). These results are first established for Toeplitz Gaussian matrices, where the system input sequence $h_i$, $-m+2\leq i \leq n$, are i.i.d. Gaussian random variables.
We then extend our results to Topelitz matrix with non-Gaussian input sequences. The main contribution of this paper is the
establishment of the performance guarantee of Toeplitz structured matrices in
parameter estimation in the presence of both outliers
and random noises. In particular, we calculated the thresholds on the
sparsity $k$ such that the parameter vector can be perfectly calculated in the
presence of a unknown outlier vector
with no more than $k$ non-zero elements
using (\ref{eq:errorcorrection}).  When both outliers and observation noise
are
present, we have shown that the estimation error goes to $0$ asymptotically as
long as
the probability density function for observation noise is not ``vanishing''
around $0$.

We also like to point out that
there is a well known duality between compressed sensing
\cite{Neighborlypolytope,DonohoTanner} and sparse error detection
\cite{CandesErrorCorrection,CT1}: the null space of sensing matrices in
compressed sensing corresponds to the tall matrix $H$ in sparse error
corrections. Toeplitz and circulant matrices have been studied in compressed
sensing in several papers
\cite{BajwaToeplitz,randomconvolution,Rauhut}.
In these papers, it has been shown that Toeplitz matrices are good for recovering
sparse vectors from undersampled measurements. In contrast, in our model of
parameter estimation in the presence of outliers, the signal itself is \emph{not}
sparse and the linear system involved is overdetermined rather underdetermined.
Also, the null space of a Toeplitz matrix does not necessarily
correspond to another Toeplitz matrix; so the problem studied in this paper is
essentially different from those studied
in \cite{BajwaToeplitz,randomconvolution,Rauhut}.

The rest of this paper is organized as follows.  In Section
\ref{sec:strongthreshold}, we derive performance bounds on the number of outliers
we can correct when only outliers are present. In Section
\ref{sec:noiseerrortogether}, we derive the estimation of system parameters
when both outliers and random noises are present. In Section \ref{sec:nonGaussian}, we extend our results to non-Gaussian inputs in system identification.  In Section \ref{sec:numerical}, we provide the numerical results and in Section \ref{sec:conclusion} conclude our paper
by discussing extensions and future directions.

\section{With Only Outliers}
\label{sec:strongthreshold}
We establish one main result regarding the threshold of successful recovery of
the system parameter vector $x$ by $\ell_1$-minimization.
\begin{theorem}
\label{thm:strongmain}
Let $H$ be an $n \times m$ Toeplitz Gaussian matrix as in (\ref{Top}), where $m$ is a
fixed positive integer and $h_i$, $-m+2 \leq i \leq n$ are i.i.d. $N(0,1)$
Gaussian random variables. Suppose that $\y=H\x+\e$, where $\e$ is a sparse
outlier vector with no more than $k$ non-zero elements. Then there exists a
constant $c_1 >0$ and a constant $\beta>0$ such that, with probability
$1-e^{-c_1n}$ as $n \rightarrow \infty$, the $n \times m$ Toeplitz matrix
$H$  has the following property.

For every $\x \in R^m$ and every error $\e$ with
its support $K$ satisfying $|K|=k \leq \beta n$, $\x$ is the unique solution
to (\ref{eq:errorcorrection}). Here the constant $0<\beta<1$ can be taken as
any number such that for some constant $\mu>0$ and $0<\delta<1$,
$$
\beta \log(1/\beta)+(1-\beta) \log(\frac{1}{1-\beta})+m\beta [\log(2)+\frac{m\mu^2
}{2}
$$
$$
+ \log(\Phi(\mu \sqrt{m}))]+(\frac{1}{2m-1}-\beta) [\log(2)
$$
$$
+
\frac{1}{2}\mu^2(1-\delta)^2+\log(1-\Phi(\mu(1-\delta)))]<0
$$
where $\Phi(t)=\frac{1}{\sqrt{2\pi}}\int_{-\infty}^{t}{e^{-\frac{x^2}{2}}\,dx}$
is the cumulative distribution function for the standard Gaussian random variable.
\end{theorem}

\textbf{Remark}: The derived correctable fraction of errors $\beta$ depends on the system dimension $m$.

We first outline the overall derivation strategy, and then go on to prove Theorem
\ref{thm:strongmain}.

Our derivation is based on checking the following
now-well-known theorem for $\ell_1$ minimization (see \cite{Yin}, for example).
\begin{theorem}
 (\ref{eq:errorcorrection}) can recover the parameter vector $\x$ exactly
whenever $\|\e\|_0 \leq k$, if and only if for every vector $z \in R^{m}\neq 0$,
$\|(Hz)_K\|_{1} < \|(Hz)_{\overline{K}}\|_{1}$ for every subset
$K \subseteq \{1,2,...,n\}$ with cardinality $|K|=k$,
where $\overline{K}=\{1,2,...,n\}\setminus K$.
\label{thm:balanced}
\end{theorem}

The difficulty of checking this condition is that the elements of $H$ are not independent random variables and that the condition must hold for every vector in the subspace generated by $H$. We adopt the following strategy of discretizing the subspace generated by $H$,see \cite{DMTSTOC,StojnicXuHassibi,Ledoux01}.
%
It is obvious that we only need to consider $Hz$ for $z \in R^{m}$ with $\|z\|_2=1$. We then pick a finite set $V=\{v_1, ..., v_{N}\}$ called $\gamma$-net on $\{z|\|z\|_2=1\}$ for a constant $\gamma>0$: in a $\gamma$-net, for every point $z$ from $\{z|\|z\|_2=1\}$, there is a $v_l \in V$ such that $\|z-v_l\|_2 \leq \gamma$. We subsequently establish the property in Theorem \ref{thm:balanced} for all the points in $\gamma$-net $V$, before extending the property in Theorem \ref{thm:balanced} to $Hz$ with $\|z\|_2=1$ not necessarily belonging to the $\gamma$-net.

Following this method, we divide the derivation into Lemmas \ref{cor:s1},
\ref{lemma:blconcentration} and \ref{lemma:slp}, which lead to Theorem \ref{thm:strongmain}.

We now start to derive the results in Theorem \ref{thm:strongmain}.  We first show the concentration of measure phenomenon for $Hz$, where $z \in R^m$ is a single vector with $\|z\|_2=1$, in Lemmas \ref{cor:s1},
\ref{lemma:blconcentration}.

\begin{lemma}\label{cor:s1}
Let $\|z\|_2=1$. For any $\epsilon >0$, when $n$ is large enough, with probability at least $1-2e^{-\frac{n^2}{(n+m-1)m}}$, it
holds
that
$$
(1-\epsilon)S \leq \|Hz\|_1\leq (1+\epsilon)S
$$
where $S=nE\{|X|\}$ and $X$ is a random variable following the Gaussian
distribution $N(0,1)$. Namely there exists a constant $c_2>0$, such that
$$
(1-\epsilon)S \leq \|Hz\|_1\leq (1+\epsilon)S
$$
holds with probability $1-2e^{-c_2 n}$ as $n \rightarrow \infty$.
\end{lemma}

\begin{proof}:
Our derivations rely on the concentration of
measure inequalities and the Chernoff bounds for Gaussian random variables
\cite{Ledoux01,LedouxTalagrand1991} .

  \begin{proposition}(Gaussian concentration inequality for Lipschitz functions)
  Let $f: R^d\rightarrow R$ be a function which is Lipschitz with constant $L$
(i.e. for all $a \in R^d$ and $b \in R^d$, $|f(a)-f(b)|\leq L\|a-b\|_2$).
Then for any $t$, we have
  \begin{equation*}
  P(|f(X)-E\{f(X)\}|\geq t) \leq 2e^{-\frac{t^2}{2L^2}},
  \end{equation*}
  where $X$ is a vector of $d$ i.i.d. standard Gaussian random variables $N(0,1)$.
  \end{proposition}

We show that, for any $\|z\|_2=1$, the function $f(h)=\|Hz\|_1$ is a function of
Lipschitz constant $\sqrt{m(n+m-1)}$, where $h=(h_{-m+2}, h_{-m+1}, ..., h_n)$.
  For two vectors $h^1$ and $h^2$, by the triangular inequality for $\ell_1$ norm
and the Cauchy-Shwarz inequality,
  \begin{eqnarray*}
  |f(h^1)-f(h^2)|&\leq& \sum_{i=-m+2}^{n}|(h^1)_i-(h^2)_i| \times \|z\|_1\\
   &\leq&\sqrt{n+m-1} \|h^1-h^2\|_2 \sqrt{m} \|z\|_2\\
   &=&\sqrt{(n+m-1)m}\|h^1-h^2\|_2,
  \end{eqnarray*}
since $\|z\|_2=1$. Then a direct application of the Gaussian concentration
inequality above leads us to Lemma \ref{cor:s1}.
  \end{proof}

\begin{lemma}
Let $\|z\|_2=1$ and $0<\delta<1$ be a constant. Then there exists a threshold
$\beta \in (0,1)$ and a constant $c_3>0$ (depending on $m$ and $\beta$), such
that, with a probability $1-e^{-c_3n}$, for all subsets $K\subseteq \{1,2,...,n\}$
with cardinality $\frac{|K|}{n}\leq \beta$,
\begin{equation*}
\|(Hz)_K\|_1 \leq \frac{1-\delta}{2-\delta} \|Hz\|_1.
\end{equation*}
\label{lemma:blconcentration}
\end{lemma}

\begin{proof}:
  Note that for a vector $z$ from the unit Euclidean norm in $R^m$, we have
$\sum\limits_{i \in K}   {|(Hz)_i|} \leq \sum\limits_{i \in K} \sum\limits_{j=1}^{m}|H_{i,j}||z_j|$,
and by the Cauchy-Schwarz   inequality,
  \begin{equation*}
  \sum_{i \in K}\sum_{j=1}^{m}|H_{i,j}||z_j| \leq \sum_{j \in J} \sqrt{m} |h_j|,
  \end{equation*}
  where $h_j$ is an element of $h$ and $J \subseteq \{-m+2, ...,n\}$ is the set of
indices   $j$ such that $h_j$ is involved in $H_{K}$. So the cardinality
$|J| \leq mk$. By the definiiton of Toeplitz matrices, the number of rows in $H$ that involve only $h_j$'s with $j$ coming from $\{-m+2, ...,n\} \setminus J$,
is at least $n-k\times(2m-1)$. Among these rows involving only $h_j$'s with $j$ coming from $\{-m+2, ...,n\} \setminus J$, one can pick at least $\frac{n-k\times(2m-1)}{2m-1}$ rows, such that these rows involve distinct $h_j$'s from each other.

Thus for a fixed vector $z$, there exists a set
$I \subseteq \{1,2,...,n\}\setminus K$ with cardinality at least
$\frac{n-k\times(2m-1)}{2m-1}$ such that $(Hz)_i$, $i \in I$, are independent
$N(0,1)$ Gaussian variables; moreover, these $(Hz)_i$, $i \in I$, are independent
from those $h_j$'s with $j \in J$.
  Thus for a fixed set $K$, the probability that $\|(Hz)_K\|_1 >
\frac{1-\delta}{2-  \delta} \|Hz\|_1$ is smaller than the probability that
  \begin{equation*}
  \sqrt{m}\sum_{i=1}^{mk} |h_i| \geq (1-\delta)
\sum_{j=1}^{\frac{n-k\times(2m-1)}{2m-1}}|h_  j|,
  \end{equation*}
  where $h_i$'s and $h_j$'s are all i.i.d. $N(0,1)$ Gaussian random variables.

  Now we use the Chernoff bound,
  \begin{eqnarray*}
  &&P(\sqrt{m}\sum_{i=1}^{mk} |h_i| \geq (1-\delta)\sum_{j=1}^{\frac{n-k\times(2m-1)}{2m-1}  }|h_j| )   \\
  &\leq& \min_{\mu \geq 0}E\left\{e^{\mu(\sqrt{m}\sum\limits_{i=1}^{mk} |h_i| - (1-\delta)\sum\limits_{j= 1}^{\frac{n-k\times(2m-1)}{2m-1}} |h_j|)}\right\}\\
  &=& \min_{\mu \geq 0}E\left\{ e^{\mu(\sqrt{m}\sum\limits_{i=1}^{mk} |h_i|)} \right\} E\left\{ e^  {- \mu(1-\delta)\sum\limits_{i=1}^{\frac{n}{2m-1}-k} |h_i|}\right\}
  \end{eqnarray*}

  After simple algebra, we have
$$
E\left\{ e^{\mu(\sqrt{m} |h_i|)} \right\}
= 2e^{\frac{\mu^2   m}{2}} F(\mu\sqrt{m})
$$
and
  $$
E\left\{ e^{-\mu(1-\delta)|h_i|} \right\}= 2e^{\frac{\mu^2(1-\delta)^2}{2}}
(1-F(\mu(1-  \delta)))
$$
where $F(t)$ is the cumulative distribution function for a standard Gaussian
random variable $N(0,1)$.

  Putting this back into the Chernoff bound and notice that
  there are at most $\binom{n}{k}$ possible support sets $K$ with cardinality
$k$, the probability $P$ that $\sqrt{m}\sum\limits_{i=1}^{mk} |h_i| \geq
\beta\sum\limits_{i=1}^{\frac{n-k\times(2m-1)}{2m-1}}|h_i|$ is violated for at least
one support set $K$ is upper bounded by the union bound
  \begin{eqnarray*}
&& \log\left(\binom{n}{k}\right)+ mk[\log(2)+\frac{m\mu^2}{2}+
\log(F(\mu\sqrt{m}))]\\
&&+(\frac{n}{2m-1}-k) [\log(2)+ \frac{\mu^2(1-\delta)^2}{2}+\log(Q((1-\delta)u))],
\end{eqnarray*}
where $Q((1-\delta)u)=1-F((1-\delta)u)$ is the Gaussian tail function.

Let $k=\beta n$, and thus $\log(\binom{n}{k})/n \rightarrow H(\beta)$ as
$n \rightarrow \infty$, where $H(\beta)=\beta \log(1/\beta)+(1-\beta)
\log(\frac{1}{1-\beta})$ is the entropy function. As long as, for a certain $\mu>0$,
\begin{eqnarray*}
&&H(\beta)+m\beta \times [\log(2)+\frac{m\mu^2}{2}+ \log(F(\mu \sqrt{m}))]+\\
&&(\frac{1}{2m-1}-\beta)[\log(2)+\frac{\mu^2(1-\delta)^2}{2}+\log(Q(\mu(1-\delta)))]\\
&&<0,
\end{eqnarray*}
then $\beta$ is within the correctable error region for all the support sets
$K$ with high probability. In fact, the last quantity can always be made
smaller than $0$ if we take $\beta$ small enough. To see this, we can first take $\mu$ large enough, such that
$[\log(2)+\frac{1}{2}\mu^2(1-\delta)^2 +\log(1-F(\mu(1-\delta)))]<0$. This is always doable, because the tail function $1-F(\mu(1-\delta))$
satisfies
$$ 1-F(\mu(1-\delta)) \leq \frac{\frac{1}{\sqrt{2\pi}}e^{-\frac{\mu^2(1-\delta)^2}{2}}}{\mu(1-\delta)}$$
for $\mu(1-\delta)>0$.
Fixing that $\mu$, we can then take $\beta$ sufficiently small such that $\log(P)<0$, as $n \rightarrow \infty$.
\end{proof}

We have so far only considered the condition in Theorem \ref{thm:balanced} for a single point $z$ on the $\gamma$-net. By a union bound on the size of $\gamma$-net, Lemma \ref{cor:s1} and \ref{lemma:blconcentration} indicate that, with overwhelming
probability, the recovery condition in Theorem \ref{thm:balanced} holds for the discrete points on $\gamma$-net.
The following lemma formally proves this fact, and then extends the proof of the recovery condition for every point in the set $\{z|\|z\|_2=1\}$.
\begin{lemma}\label{lemma:slp}
There exist a constant $c_4>0$ such that when $n$ is large enough, with probability $1-e^{-c_4n}$, the Gaussian Toeplitz matrix $H$ has the following property: for every $z \in R^m$ and every subset $K
\subseteq \{1,...,n\}$ with $|K| \leq \beta n$, $\sum \limits_ {i \in
\overline{K}} |(Hz)_i| - \sum \limits_{i \in K} |(Hz)_i| \geq \delta' S
$, where $\delta'>0$ is a constant.
\end{lemma}
\begin{proof} For any given $\gamma>0$, there exists a $ \gamma$-net $V=\{v_1, ..., v_{N}\}$ of
cardinality less than $(1+\frac{2}{\gamma})^m$\cite{Ledoux01}. Since each row of $H$ has $m$ i.i.d $N(0,1)$ entries, elements of $Hv_j$, $1\leq j\leq N $, are
(not independent) $N(0,1)$ entries. Applying a union bound on the size of $\gamma$-net, Lemmas \ref{lemma:blconcentration}
and \ref{cor:s1} imply that for every $v_j \in V$, for some $\delta>0$ and for any constant
$\epsilon>0$, with probability $1-2e^{-cn}$ for some $c>0$,
$$\|(Hv_j)_K\|_1 \leq \frac{(1-\delta)(1+\epsilon)}{2-\delta}S$$
$$(1-\epsilon)S \leq \|Hv_j\|_1\leq (1+\epsilon)S$$
hold simultaneously for every vector $v_j$ in $V$.

For any $z$ such that $\|z\|_2=1$, there exists a point $v_0$ (we change the subscript numbering for $V$ to index the order) in $V$ such
that $\|z-v_0\|_2\triangleq \gamma_1 \leq \gamma$. Let $z_1$ denote $z-v_0$,
 then $\|z_1-\gamma_1v_1\|_2 \triangleq \gamma_2 \leq \gamma_1 \gamma \leq \gamma^2$ for
 some $v_1$ in $V$. Repeating this process, we have $z=\sum_{j\geq 0} \gamma_j v_j$, where $\gamma_0=1$, $\gamma_j \leq \gamma^j$ and $v_j \in V$.

Thus for any $z \in R^m$, $z=\|z\|_2\sum_{j\geq 0} \gamma_j v_j$. For any index set $K$ with $|K| \leq \beta n$,
\begin{eqnarray*}
\sum \limits_{i \in K} |(Hz)_i| &=& \|z\|_2 \sum \limits_{i \in K} |(\sum \limits_{j \geq 0} \gamma_j Hv_j)_i| \\
& \leq & \|z\|_2 \sum \limits_{i \in K} \sum \limits_{j \geq 0} \gamma^{j} |(Hv_j)_i| \\
&= & \|z\|_2  \sum \limits_{j \geq 0} \gamma^{j} \sum \limits_{i \in K} |(Hv_j)_i| \\
& \leq & S\|z\|_2 \frac{(1-\delta)(1+\epsilon)}{(2-\delta){(1-\gamma)}}
\end{eqnarray*}

\begin{eqnarray*}
\sum \limits_{i} |(Hz)_i| &=& \|z\|_2 \sum \limits_{i } |(\sum \limits_{j \geq 0} \gamma_{j} Hv_j)_i| \\
& \geq & \|z\|_2 \sum \limits_{i}(|(Hv_0)_i|- \sum \limits_{j \geq 1} \gamma_{j} |(Hv_j)_i|) \\
& \geq & \|z\|_2 (\sum \limits_{i} |(Hv_0)_i|- \sum \limits_{j \geq 1} \gamma^{j} \sum \limits_{i}|(Hv_j)_i|) \\
& \geq & \|z\|_2 ((1-\epsilon)S- \sum \limits_{j \geq 1} \gamma^{j} (1+\epsilon)S) \\
& \geq &  S\|z\|_2 (1-\epsilon- \frac{\gamma(1+\epsilon)}{1-\gamma}).
\end{eqnarray*}

So $\sum \limits_ {i \in \overline{K}} |(Hz)_i| - \sum \limits_{i \in K}
|(Hz)_i| \geq S \|z\|_2
( 1-\epsilon- \frac{\gamma(1+\epsilon)}{1-\gamma}-2\frac{(1-\delta)(1+\epsilon)}{(2-\delta){(1-\gamma)}})$.
For a given $\delta$, we can pick $\gamma$ and $\epsilon$ small
enough such that $\sum \limits_ {i \in \overline{K}} |(Hz)_i| - \sum
\limits_{i \in K} |(Hz)_i| \geq \delta' S \|z\|_2$, satisfying the condition in Theorem \ref{thm:balanced}.
\end{proof}

\begin{proof} (of Theorem \ref{thm:strongmain}). A direct consequence of Theorem \ref{thm:balanced} and Lemmas
\ref{lemma:blconcentration}, \ref{cor:s1}, and \ref{lemma:slp}.
\end{proof}

So far, we have considered the \emph{uniform} performance guarantee of $\ell_1$ minimization for all the sets with cardinality up to $k$.
Instead, given a support set $K$ for the outliers (though we do not know what the support set is before performing $\ell_1$ minimization), the correctable sparsity $\frac{|K|}{n}$ of outliers can go to $1$.  This result is formally stated in Theorem \ref{thm:weakthreshold_pureoutlier}.
\begin{theorem}
Take an arbitrary constant $0<\beta<1$ and let $\y=H\x+\e$, where $H$ is a
Gaussian Toeplitz matrix, and $\e$ is \emph{an}
outlier vector with no more than $k=\beta n$ non-zero elements. When
$n \rightarrow \infty$, $\x$ can be recovered perfectly using $\ell_1$
minimization from $\e$ with high probability.
\label{thm:weakthreshold_pureoutlier}
\end{theorem}
\begin{proof}
The statement follows from Theorem \ref{thm:noiseerrortogether}, specialized to the case $\w=0$.
\end{proof}

%
\section{With Both Outliers and Observation Noises}
\label{sec:noiseerrortogether}
We further consider Toeplitz matrix based system identification when  both
outliers and random observation errors are present, namely, the observation $\y=H\x+\e+\w$,
where $\e$ is a sparse outlier vector with no more than $k$ non-zero elements,
and $\w$ is the vector of additive observation noises. We show that, under mild conditions, the identification
error $\|\hat{\x}-\x\|_2$ goes to $0$ even when there are both outliers and
random observation errors, where $\hat{\x}$ is the solution to (\ref{eq:errorcorrection}).
\begin{theorem}
\label{thm:noiseerrortogether}
Let $m$ be a fixed positive integer and $H$ be an $n \times m$ Toeplitz matrix ($m<n$)
in (\ref{Top}) with each element $h_i$, $-m+2\leq i \leq n$, being i.i.d. $N(0,1)$ Gaussian random variables. Suppose $$\y=H\x+\e+\w,$$ where $\e$ is a sparse
vector with $k \leq \beta n$ non-zero elements ($\beta<1$ is a constant) and $\w$ is
the observation noise vector. For any constant $t>0$, we assume that, with high probability as $n \rightarrow \infty$,  at least $\alpha(t)n$ (where $\alpha(t)>0$ is a constant depending on $t$ ) elements in $\w+\e$ are no bigger than $t$ in amplitude.

Then $\|\hat{\x}-\x\|_2\rightarrow 0$ with high probability as $n \rightarrow \infty$, where $\hat{\x}$ is the solution to (\ref{eq:errorcorrection}).
\end{theorem}

We remark that, in Theorem \ref{thm:noiseerrortogether}, the condition  on
the unknown outlier vector is merely $\beta <1$, and the condition on the random noise $\w$ is weaker than the usual condition of having i.i.d. elements with median $0$ \cite{SIC}. In fact, if $\w$ is independent from $\e$, and the elements of $\w$ are i.i.d. random variables following a distribution which is not ``vanishing'' in an arbitrarily small region around $0$ (namely the cumulative distribution function $F(t)$ satisfies that $F(t)-F(-t)>0$ for any $t>0$.
Note that the probability density function $f(t)$ is allowed to be $0$,however), the conditions in Theorem \ref{thm:noiseerrortogether} will be satisfied.

To see that, first observe that $(1-\beta)n$ elements of the outlier
vector are zero. If elements of $\w$ are i.i.d.
following a probability density function $f(s)$ that is not ``vanishing''
around $s=0$ , with probability converging to one
as $n \rightarrow \infty$, at least $[F(t)-F(-t)](1-\beta)(1-\epsilon)n=\alpha(t)n$ elements of
the vector $e+\w$ are no bigger than $t$, where $\epsilon>0$ is an arbitrarily small number.
Gaussian distributions, exponential
distributions, Gamma distributions and many other distributions for $\w$ all
satisfy such conditions in
Theorem \ref{thm:noiseerrortogether}. This greatly broadens the
existing results, e.g., in \cite{SIC}, which requires $f(0)>0$ and does not
accommodate outliers. Compared with analysis in compressed sensing
\cite{CandesErrorCorrection,DonohoNoise}, this result is for Toeplitz matrix in
system identification and applies to observation noises with non-Gaussian
distributions. The results in this paper also improve on the performance bounds in \cite{CDC2011}, by showing that the identification error actually goes to $0$.

\begin{proof} (of Theorem \ref{thm:noiseerrortogether})
$\|\y-H\hat{\x}\|_1$ can be written as $\|H(\x-\hat{\x})+\e+\w\|_1$. We argue that for any constant $t>0$, with high probability as $n \rightarrow \infty$, for all $\hat{\x}$ such that $\|\x-\hat{\x}\|=t$, $\|H(\x-\hat{\x})+\e+\w\|_1 > \|\e+\w\|_1$, contradicting the assumption that $\hat{\x}$ is the solution to (\ref{eq:errorcorrection}).

To see this, we cover the sphere $Z=\{z| \|z\|_2=1\}$ with a $\gamma$-net $V$. We first argue that for every discrete point $tv_j$ with $v_j$ from the $\gamma$-net, $\|H tv_j+\e+\w\|_1 > \|\e+\w\|_1$; and then extend the result to the set $tZ$.

Let us denote
\begin{eqnarray*}
g(h,t)&=&\|H tv_j+\e+\w\|_1 - \|\e+\w\|_1\\
&=&\sum_{i=1}^{n}(|l_i+t (Hv_j)_i|-|l_i|),\end{eqnarray*}
where $l_i=(\e+\w)_i$ for $1\leq i \leq n$. We note that $(Hv_j)_i$ is a Gaussian random variable $N(0,1)$. Let $X$ be a Gaussian random variable $N(0,\sigma^2)$, then for an arbitrary number $l$,
\begin{eqnarray*}
&&E\left\{ {|l+tX|}-|l| \right\}\\
&=&\frac{2}{\sqrt{2\pi}t\sigma}\int_{0}^{\infty} x e^{-\frac{(|l|+x)^2}{2t^2\sigma^2}}\,dx\\
&=& \sqrt{\frac{2}{\pi}}t\sigma e^{-\frac{l^2}{2t^2\sigma^2}}-2|l|(1-\Phi(\frac{|l|}{t\sigma})),
\end{eqnarray*}
which is a decreasing nonnegative function in $|l|$. From this, $E\{g(h,t)\}=\sum_{i=1}^{n}(\sqrt{\frac{2}{\pi}}t e^{-\frac{|l_i|^2}{2t^2\sigma^2}}-2|l_i|(1-\Phi(\frac{|l_i|}{t})))$.
When $|l|\leq t$ and $\sigma=1$,
\begin{eqnarray*}
&~&E\left\{ {|l+tX|}-|l| \right\}\\
&=&\sqrt{\frac{2}{\pi}}t e^{-\frac{1}{2}}-2|l|(1-\Phi(1))\\
&\geq& 0.1666t.
\end{eqnarray*}
It is not hard to verify that $|g(a,t)-g(b,t)|\leq \sum_{i=1}^{n}t \sqrt{m}|a_i-b_i| \leq t\sqrt{mn}\|a_i-b_i\|_2$, and $g(h,t)$ has a Lipschitz constant (for $h$) no bigger than $t\sqrt{mn}$. Then by the concentration of measure phenomenon for Gaussian random variables (see \cite{Ledoux01,LedouxTalagrand1991}),
\begin{eqnarray*}
&&P(g(h,t)\leq 0)\\
&=& P(\frac{g(h,t)-E\{g(h,t)\}}{t\sqrt{mn}} \leq -\frac{E\{g(h,t)\}}{t\sqrt{mn}} )\\
&\leq& 2e^{-\frac{ \left(\sum_{i=1}^{n}\left[\sqrt{\frac{2}{\pi}}t e^{-\frac{l_i^2}{2t^2}}-2|l_i|(1-\Phi(\frac{|l_i|}{t}))\right]\right)^2  }{2t^2 nm}}\\
&\triangleq& 2e^{-B}.
\end{eqnarray*}
If there exists a constant $\alpha (t)$ such that, as $n \rightarrow \infty$, at least $\alpha(t) n$ elements have magnitudes smaller than $t$, then the numerator in $B$ behaves as $\Theta(n^2)$ and the corresponding probability $P(g(h,t)\leq 0)$ behaves as $2e^{-\Theta( n)}$. This is because when $|l|\leq t$, $\sqrt{\frac{2}{\pi}}t e^{-\frac{|l|^2}{2t^2}}-2|l|(1-\Phi(\frac{|l|}{t})) \geq 0.1666 t$.

By the same reasoning, $g(h,t)\leq \epsilon n$ holds with probability no more than $e^{-c_5 n}$ for each discrete point from the $\gamma$-net $tV$, where $\epsilon>0$ is a sufficiently small constant and $c_5>0$ is a constant which may only depend on $\epsilon$.  Since there are at most $(1+\frac{2}{\gamma})^{m}$ points from the $\gamma$-net, by a simple union bound, with probability $1-e^{-c_6 n}$ as $n \rightarrow \infty$, $g(h,t)> \epsilon n$ holds for \emph{all} points from the $\gamma$-net $tV$, where $c_6>0$ is a constant and $\gamma$ can be taken as an arbitrarily small \emph{constant}. Following similar $\gamma$-net proof techniques for Lemmas \ref{cor:s1}, \ref{lemma:blconcentration} and \ref{lemma:slp}, if we choose a sufficiently small constant $\epsilon>0$ and accordingly a sufficiently small constant $\gamma>0$, $g(h,t)> 0.5 \epsilon n$ holds simultaneously for every point in the set $tZ$ with high probability $1-e^{-c_7 n}$, where $c_7>0$ is a constant.

Notice if $g(h, t)>0$ for $t=t_1$, then necessarily $g(h,t)>0$ for $t=t_2>t_1$. This is because $g(h,t)$ is a convex function in $t\geq 0$ and $g(h,0)=0$. So if $g(h,t)> 0.5 \epsilon n>0$ holds with high probability for every point $tZ$, necessarily $\|\hat{\x}-\x\|_2<t$, because $\hat{\x}$ minimizes the objective in (\ref{eq:errorcorrection}). Because we can pick arbitrarily small $t$, $\|\hat{\x}-\x\|_2\rightarrow 0$ with high probability as $n \rightarrow \infty$.
\end{proof}


 \section{Extensions to Non-Gaussian Inputs}
\label{sec:nonGaussian}

 In previous sections, we have considered Gaussian inputs for system
identifications. In fact, our results also extend to non-Gaussian inputs.
We illustrate this by considering a widely applied input for
system identification, the pseudo-random binary sequence (PRBS) which is a
 random
binary sequence taking values $\pm 1$ with the equal probability.
To simply our analysis, we only consider strong
performance bounds of uniformly correcting every possible set of outliers
with cardinality smaller than $k$. By noting that the PRBS input is actually
an i.i.d. Bernoulli sequence, we first state the main result.

 \begin{theorem}
\label{thm:strongmain1}
Let $H$ be an $n \times m$ Toeplitz matrix as in (\ref{Top}), where $m$ is a
fixed positive integer and $h_i$, $-m+2 \leq i \leq n$ are i.i.d. Bernoulli
random variables taking values $+1$ and $-1$ with equal probability.
Suppose that $\y=H\x+\e$, where $\e$ is a sparse
outlier vector with no more than $k$ non-zero elements.
Then there exists a constant $c_1 >0$ and a constant $\beta>0$ such that, with probability
$1-e^{-c_1n}$ as $n \rightarrow \infty$, the $n \times m$ Toeplitz matrix
$H$  has the following property: for every $\x \in R^m$ and every error $\e$ with
its support $K$ satisfying $|K|=k \leq \beta n$, $\x$ is the unique solution
to (\ref{eq:errorcorrection}).
\end{theorem}

In order to prove the results, again we only need to consider $Hz$ for
$z \in R^{m}$ with $\|z\|_2=1$. As in Section \ref{sec:strongthreshold},
we adopt the strategy of discretizing the subspace generated by $H$.
We pick a finite set $V=\{v_1, ..., v_{N}\}$ called $\gamma$-net
on $\{z|\|z\|_2=1\}$ for a constant $\gamma>0$: in a $\gamma$-net,
for every point $z$ from $\{z|\|z\|_2=1\}$, there is a $v_l \in V$
such that $\|z-v_l\|_2 \leq \gamma$. We subsequently establish the property
in Theorem \ref{thm:balanced} for the points in the $\gamma$-net $V$, before
extending the results to every point $Hz$ with $\|z\|_2=1$.

We first need a few lemmas and the concentration of measure phenomenon for $Hz$,
where $z \in R^m$ is a single
vector with $\|z\|_2=1$. To this end, we need to use the McDiarmid's inequality.
\begin{theorem} {(McDiarmid's inequality \cite{McDiarmid})}
Let ${X_1,\ldots,X_n}$ be independent random variables taking values in a set $\chi$, and let ${F: \chi \times \ldots \times \chi \rightarrow {\bf R}}$ be a function with the property that if one freezes all but the ${i^{th}}$ coordinate of ${F(x_1,\ldots,x_n)}$ for some ${1 \leq i \leq n}$, then ${F}$ only fluctuates by most ${c_i > 0}$, thus,
\begin{eqnarray*}
&& |F(x_1,\ldots,x_{i-1},x_i,x_{i+1},\ldots,x_n) -\\
&&~~~\displaystyle F(x_1,\ldots,x_{i-1},x_i',x_{i+1},\ldots,x_n)| \leq c_i
\end{eqnarray*}
    for all ${x_j \in \chi}$ and ${x'_i \in \chi}$, for ${1 \leq i, j \leq n}$. Then for any ${\lambda > 0}$, one has
$$ {\bf P}( |F(X) - {\bf E}[ F(X)]| \geq \lambda ) \leq 2 e^{-\frac{2\lambda^2}{\sigma^2}  },$$
where ${\sigma^2 := \sum_{i=1}^n c_i^2}$.
\end{theorem}

\begin{lemma}\label{cor:Bernoullis1}
Let $\|z\|_2=1$. For any $\epsilon >0$, with probability $1-2e^{-\frac{\epsilon^2 n^2}{2(n+m-1)m}}$, it
holds that
$$
S-\epsilon n \leq \|Hz\|_1\leq S+\epsilon n
$$
where $S=nE\{|X|\}$, and $X=\sum\limits_{i=1}^{m} z_i h_i$ with $h_i$'s being independent Bernoulli random variables.
\end{lemma}

\begin{proof}
Define $f(h_{-m+2},...,h_{n})=\|Hz\|_1$. We note that,  by changing the value of only one Bernoulli variable $h_i$,  $f(h_{-m+2},...,h_{n})$ changes at most by $2\|z\|_1\leq 2\sqrt{m}$. The conclusion of this lemma then follows from applying the McDiarmid's inequality.
\end{proof}

\begin{lemma}
Let $\|z\|_2=1$, $\beta \in (0,1)$ and $\epsilon$ be an arbitrarily small positive number. Let $K\subseteq \{1,2,...,n\}$ be an arbitrary subset with cardinality $\frac{|K|}{n}\leq \beta$. With a probability of at least $1-2e^{-\frac{\epsilon^2 n^2}{2(n+m-1)m}}$,
\begin{equation*}
\frac{|K|}{n}S-\epsilon n  \leq \|(Hz)_K\|_1 \leq \frac{|K|}{n}S+\epsilon n,
\end{equation*}
where $S=nE\{|X|\}$, and $X=\sum\limits_{i=1}^{m} z_i h_i$ with $h_i$'s being independent Bernoulli random variables.
\label{lemma:Bernoulliblconcentration}
\end{lemma}

\begin{proof}
Define $f(h_{-m+2},...,h_{n})=\|(Hz)_{K}\|_1$. Then again,  the conclusion follows from applying the McDiarmid's inequality.
\end{proof}

We also have simple upper and lower bounds on $S$ in Lemma \ref{cor:Bernoullis1} and Lemma \ref{lemma:Bernoulliblconcentration}.
\begin{lemma}
 For any $z$ with $\|z\|_2=1$.
 \begin{equation*}
\frac{n}{2 \sqrt{m}}\leq S \leq n \sqrt{m}
\end{equation*}
\label{lemma:Bernoulliupperlower}
\end{lemma}

\begin{proof}
Since $\|z\|_2=1$, there must exist one $1\leq j\leq m$ such that $|z_j|\geq \frac{1}{\sqrt{m}}$. So no matter what $\sum\limits_{i=1,i\neq j}^{m} z_i h_i$ is, there is always a probability of one half such that $|\sum\limits_{i=1,i\neq j}^{m} z_i h_i+z_j h_j|\geq \frac{1}{\sqrt{m}}$. This leads to the lower bound.

The upper bound is from $E\{|X|\} \leq \sum\limits_{i=1}^{m} |z_i| |h_i| \leq  \sqrt{m} $.
\end{proof}

Combining Lemma \ref{cor:Bernoullis1}, Lemma \ref{lemma:Bernoulliblconcentration} and Lemma \ref{lemma:Bernoulliupperlower}, by a simple union bound over the $\gamma$-net, we have:
\begin{lemma}
Let $\beta \in (0,1)$, and let $K\subseteq \{1,2,...,n\}$ be an arbitrary subset with cardinality $\frac{|K|}{n}\leq \beta$. With a probability of at least $1-4(1+\frac{2}{\gamma})^m e^{-\frac{\epsilon^2 n^2}{2(n+m-1)m}}$, the following two statements hold true simultaneously for all the points $z$ of the $\gamma$-net $V$:
\begin{itemize}
\item  $$\frac{n}{2 \sqrt{m}} -\epsilon n \leq \|Hz\|_1;$$
\item  $$\|(Hz)_K\|_1 \leq \beta n \sqrt{m} +\epsilon n.$$
\end{itemize}
\end{lemma}

By using the same ``zooming'' technique as in Section \ref{sec:strongthreshold}, and a simple union bound over $\binom{n}{|K|}$ possible subsets with cardinality $K$, we can get the following lemma which extends the result from the $\epsilon$-net to $\{z|\|z\|_2=1\}$.
\begin{lemma}\label{lemma:slp1}
There exist a constant $c_4>0$ and a sufficiently small constant $\beta>0$,  such that when $n$ is large enough, with probability $1-e^{-c_4n}$, the Bernoulli  Toeplitz matrix $H$ has the following property: for every $z \in R^m$ and every subset $K
\subseteq \{1,...,n\}$ with $|K| \leq \beta n$, $\sum \limits_ {i \in
\overline{K}} |(Hz)_i| - \sum \limits_{i \in K} |(Hz)_i| \geq \delta' n
$, where $\delta'>0$ is a sufficiently constant.
\end{lemma}

\begin{proof}
of Theorem \ref{thm:strongmain1}. A direct consequence of Theorem 4.2 and lemmas
4.3-4.7.
\end{proof}

\section{Numerical Evaluations}
\label{sec:numerical}
Based on Theorem \ref{thm:strongmain}, we calculate the strong thresholds in Figure \ref{fig:gaussianfixrho} for different values of $m$ by optimizing over $\mu>0$ and $\delta$. As $m$ increases, the correlation length in the matrix $H$ also increases and the corresponding correctable number of errors decreases (but the bound $\beta=\frac{k}{n}>0$ always exists).
\begin{figure}[t]
\centering
\includegraphics[width=3in, height=2.8in]{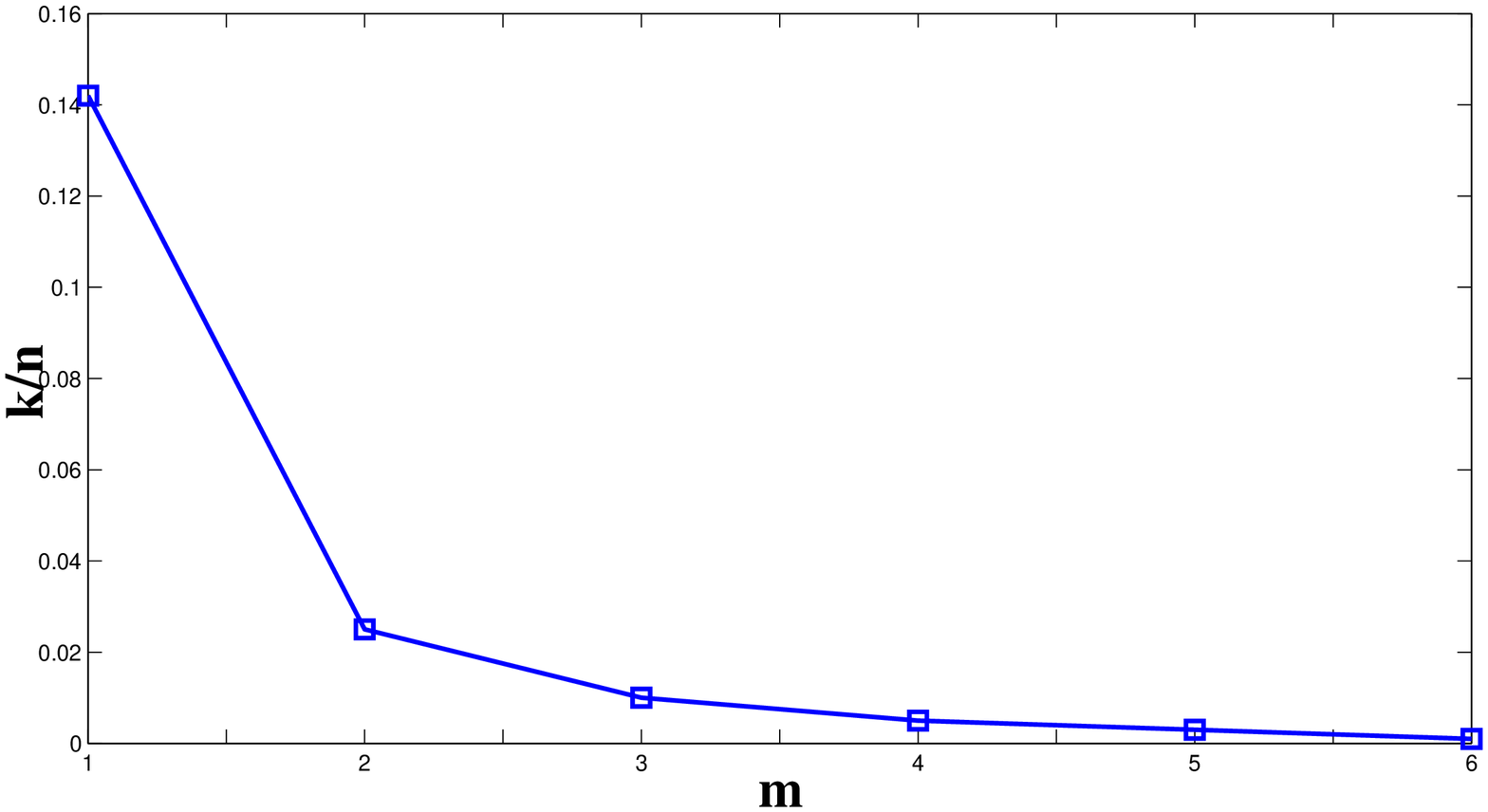}
\caption{Recoverable fraction of errors versus $m$}\label{fig:gaussianfixrho}
\end{figure}
\begin{figure}[t]
\centering
\includegraphics[width=3in, height=2.5in]{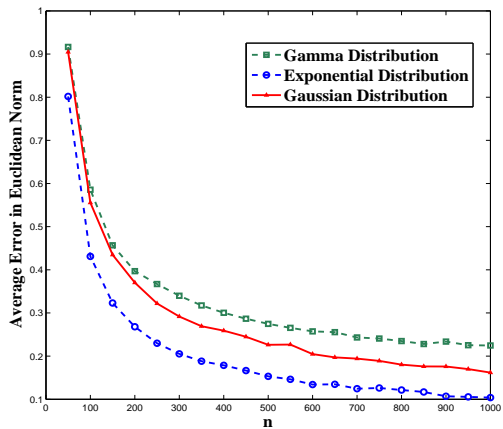}
\caption{With outliers and noises of different distributions}
\label{fig:threecurves}
\end{figure}

\begin{figure}[t]
\centering
\includegraphics[width=2.8in, height=2.8in]{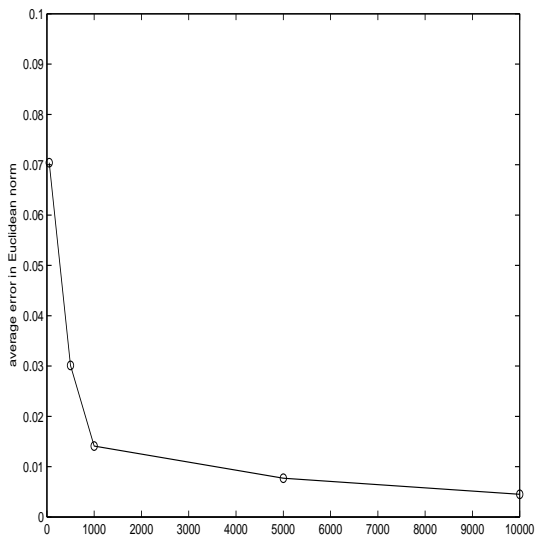}
\caption{Estimation errors }
\label{fig:together}
\end{figure}

\begin{figure}[t]
\centering
\includegraphics[width=2.8in, height=2.8in]{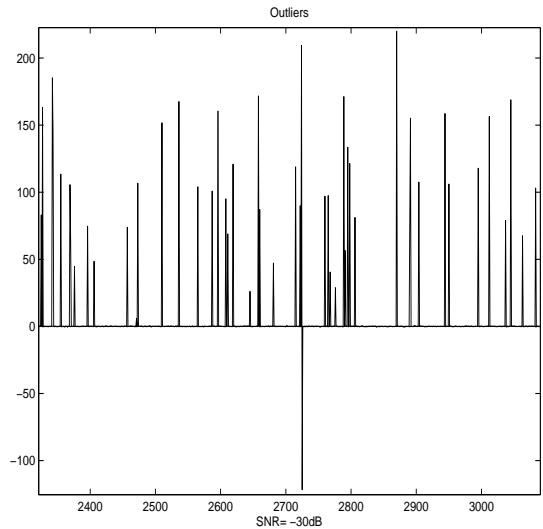}
\caption{Outliers}
\label{fig:outlier}
\end{figure}

\begin{figure}[t]
\centering
\includegraphics[width=2.8in, height=2.8in]{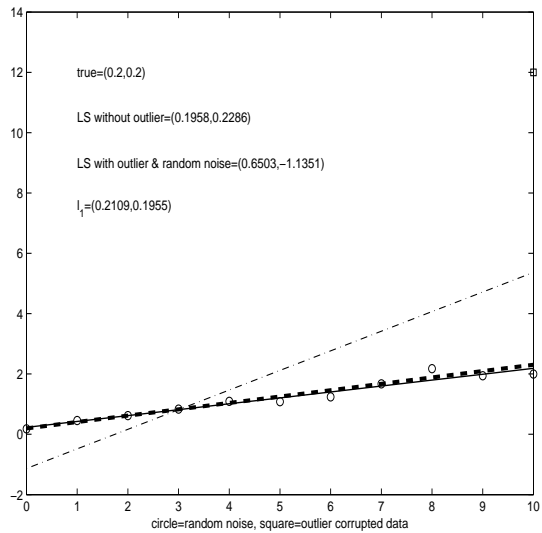}
\caption{Limited data length}
\label{fig:limited}
\end{figure}

We then evaluate in Figure \ref{fig:threecurves} the $\ell_2$-norm
error $\|\hat{\x}-\x\|_2$ of $\ell_1$ minimization for Gaussian Toeplitz
matrices under both outliers and i.i.d. observation noises of
different probability distributions: Gamma distribution with shape parameter $k=2$ and scale $\frac{1}{\sqrt{6}}$; standard Gaussian distribution $N(0,1)$; and exponential distribution with mean $\frac{\sqrt{2}}{2}$. These distributions are chosen such that the observation noises have the same expected energy. The system parameter $m$ is set to $5$ and the system state $\x$ are generated as i.i.d. standard Gaussian random variables. We randomly pick $\frac{n}{2}$ i.i.d. $N(0,100)$ Gaussian outliers with random support for the error vector $\e$. For all these distributions, the average error goes to $0$ (we also verified points beyond $n>1000$). What is interesting is that the error goes to $0$ at different rates. Actually, as implied by the proof of Theorem \ref{thm:noiseerrortogether}, for random noises with bigger probability density functions around $0$, the system identification error may be smaller.
In fact, the Gamma distribution has the worst performance because its probability density function is smaller around the origin (actually $0$ at the origin), while the exponential distribution has the best performance, because it has the largest probability density function around $0$.

 To illustrate the asymptotic convergence and
limited data length performance, we give two simulation
examples. The first example is
a five dimensional FIR system. The input $h_i$'s are an i.i.d. Gaussian
sequence of zero mean and variance $2^2$
and the random noises are i.i.d. Gaussian of zero mean and variance $0.2^2$.
The system parameter vector $\x$ is randomly generated
according to Gaussian of zero mean and unit variance.
The number of non-zero elements in the outlier vector is $rand\times0.2\times n$ where
$rand$ is the random variable uniformly in [0,1] which loosely translates
into that 10\% of the data are outliers on average.
The non-zero elements of the outliers are
randomly generated according to Gaussian of mean=100 and variance $50^2$.
The locations of the non-zeros elements are randomly assigned in each simulation.
Figure 3 shows the simulation results of the estimation error
for different $n$ which are
the average of 10 Monte Carlo runs respectively. An example of outliers
is shown in Figure 4. Since outliers are heavily biased with very large
magnitudes if non-zero, the signal to noise ratio is very small $SNR=-30dB$.
Clearly, even at this very low level of SNR,
the identification results are still very good as predicted by the analysis.
The second example is a line $y=(z,1)\begin{bmatrix}k \cr b\end{bmatrix}
+ noise$ where
the true but unknown $x=(k, b)'=(0.2,0.2)'$.
Because of noise that is an i.i.d. Gaussian of zero mean and variance $0.2^2$,
the actual data set, shown in Figure 5 as circles, is given by Table 1.
\begin{table}[t]
{\tiny
\begin{tabular}{|c|c|c|c|c|c|c|}
\hline
z & 0&1&2&3&4&5 \cr \hline
y &.1779&.4555&.6174&.8347&1.0907&1.0793 \cr \hline\hline

z& 6&7&8&9&10&~ \cr \hline
y &1.2406&1.6721&2.177&1.9386&1.9975(11.9975)&~\cr  \hline
\end{tabular}
\caption{Limited data}
}
\end{table}
First the least squares estimate $(0.1958,0.2286)'$ is calculated which is
close to the true but unknown
$x=(0.2,0.2)'$. Now, a single outlier at $z=10$ is added as a square in Figure 5.
For this single outlier, the least squares estimate becomes $(0.6503,-1.1351)'$
in dash-dot line, way off the true but unknown $x$. Finally, the $\ell_1$
optimization as discussed in the paper is applied to have
an estimate $(0.2109,0.1955)'$ in bold-dash line in Figure 5.
Virtually, the $\ell_1$ optimization
eliminates the effect of the outlier even for a very limited data case.

\section{Concluding Remarks}
\label{sec:conclusion}
We have re-visited the least absolute deviation estimator or
$\ell_1$ minimization from a compressed sensing point of view and shown that the
exact system parameter vector can be
recovered as long as the outlier is sparse in the sense that the number of
non-zero elements has to be bound by $\beta n$ for some $\beta <1$. No any
probabilistic assumption is imposed on the unknown
outlier. Further, in the presence of both the outliers and random noises, the system
parameter vector can still be recovered if some mild conditions on the
random noises are met.

\end{document}